\documentclass[a4paper,UKenglish,10pt]{scrartcl}
 
\usepackage{microtype}
\usepackage{amsthm}
\usepackage{amssymb}
\usepackage{amsmath}
\usepackage{ stmaryrd }
\usepackage{ucs,algorithmic,algorithm}
\usepackage{complexity}

\title{The arithmetic complexity of tensor contractions}

\newtheorem{theorem}{Theorem}
\newtheorem{proposition}[theorem]{Proposition}
\newtheorem{observation}[theorem]{Observation}

\newtheorem{definition}[theorem]{Definition}
\newtheorem{corollary}[theorem]{Corollary}
\newtheorem{lemma}[theorem]{Lemma}
\newtheorem{remark}[theorem]{Remark}

\newcommand{\calT}{{\mathcal{T}}}

\newcommand{\lab}[1]{\mathsf{label}({#1})}
\newcommand{\leaf}[1]{\mathsf{leaf}({#1})}
\newcommand{\maxorder}[1]{\mathsf{maxorder}({#1})}

\newcommand{\VPw}{\VP_{\mathrm{ws}}}

\newcommand{\set}[1]{[#1]}
\newcommand{\tuple}[3][1]{#2_{#1}, \ldots, #2_{#3}}
 
\newcommand{\N}{\mathbb{N}}
\renewcommand{\K}{\mathbb{K}} 
\newcommand{\e}{{\bf e}}
\newcommand{\f}{{\bf f}}

\newcommand{\sform}{\{*\}\text{-formula}}
\newcommand{\sijform}{\{*_{i,j}\}\text{-formula}}

\DeclareMathOperator{\maxdim}{\mathop{maxdim}}

\newcommand{\BIGOP}[1]{\mathop{\mathchoice%
{\raise-0.22em\hbox{\huge $#1$}}%
{\raise-0.05em\hbox{\Large $#1$}}{\hbox{\large $#1$}}{#1}}}

\begin{document}

\author{
Florent Capelli\\
ENS Lyon\\
Lyon, France\\
{\small \texttt{florent.capelli@ens-lyon.fr}}\\
\and
Arnaud Durand\\ 
IMJ UMR 7586  -  Logique\\
Universit\'{e} Paris Diderot\\
F-75205 Paris, France  \\
 {\small \texttt{durand@logique.jussieu.fr}}
\and
Stefan Mengel\thanks{Partially supported by DFG grants BU 1371/2-2 and BU 1371/3-1 and the Initial Training Network in Mathematical Logic MALOA PITN-GA-2009-238381.}\\Institute of Mathematics\\ University of Paderborn\\ D-33098 Paderborn, Germany\\ {\small\texttt{smengel@mail.uni-paderborn.de}} 
}

\maketitle

\begin{abstract}
We investigate the algebraic complexity of tensor calulus. We consider a generalization of iterated matrix product to tensors and show that the resulting formulas exactly capture $\VP$, the class of polynomial families efficiently computable by arithmetic circuits. This gives a natural and robust characterization of this complexity class that despite its naturalness is not very well understood so far.
\end{abstract}

\section{Introduction}

The question of which polynomials can be computed by polynomial size arithmetic circuits is one of the central questions of algebraic complexity. It was first brought up explicitly by Valiant \cite{val79} who formulated a complexity theory in this setting with its own complexity classes and notions of completeness. Efficient computation in Valiant's model is formalized by the complexity class $\VP$ which consists of families of polynomials that can be computed by arithmetic circuits of polynomial size. Despite recent efforts relating $\VP$ to logically defined classes of polynomial families \cite{men,ouroldpaper}, this class is not very well understood. This is reflected in the low number of helpful alternative characterizations and the conspicuous absence of any known natural complete problem.

Consequently, most progress in arithmetic circuit complexity has not been achieved by considering arithmetic circuits directly, but instead by considering the somewhat more restricted model of arithmetic branching programs (see e.g. \cite{malodportier,val79,nisan1991,koiran2012}). Arithmetic branching programs are widely conjectured to have expressivity strictly between that of arithmetic formulas and circuits, but have so far been better to handle with known proof techniques. One of the nice properties of branching programs that has often played a crucial role is that they can equivalently be seen as computing a specified entry of the iterated product of a polynomial number of matrices.

We extend this view on branching programs by going from matrices to higher dimensional tensors. Consequently, we also go from matrix product to the the generalized notion of contraction of tensors. It turns out that generalizing iterated matrix product to iterated tensor contractions does increase the expressive power of the model and that the resulting tensor formulas capture exactly $\VP$. This characterization of $\VP$ turns out to be fairly robust in the sense that one can add different restrictions on the dimensions of the tensors without changing the expressive power of the model at all.

This is not the first time that the complexity of tensor calculus is
studied. Damm, Holzer and  McKenzie\cite{mck} have characterized
different boolean complexity classes by formulas having matrices as inputs and using
addition, matrix product and tensor product as operations. Malod \cite{mal} adapted these
formulas to the arithmetic circuit setting and showed characterizations for most arithmetic circuit classes.
One difference between
these results and those in our paper is that in \cite{mck} and \cite{mal} tensors are always encoded as matrices, i.e.\ the tensor product is expressed as the Kronecker product of two matrices. 
Another difference is that both the characterization of $\VP$ obtained in \cite{mal} and the similar characterization of $\LOGCFL$ (the Boolean analogon of $\VP$) from \cite{mck} require an additional restriction, called \textit{tameness} on the size of matrices computed at each gate of the formula. This restriction permits to control the growth of the intermediate objects in the computation but may seem not very natural. In this present  work, working directly with tensors instead of a matrix representation makes such a unnecessary and a more direct connection between $\VP$ and tensor calculus is established.

The paper is organized in three parts: We first give the definitions and properties of the objects we will deal with necessary to understand the remainder of the paper. In the second part we prove the lower bound of our characterization, i.e.\ we show our tensor formulas can efficiently simulate arithmetic circuits. In the third part we prove the corresponding upper bound: We first study a very restrictive class of tensor formulas and show that polynomials computed by them can also be computed by arithmetic circuits. Then we extend this result in several steps to general tensor formulas of polynomial size.

\section{Preliminaries}
In the following, $\K$ is a field and bold letters denote tuples when
there is no ambiguity on their length.

\subsection{Arithmetic circuits}

We will use
 the well known model of arithmetic circuits to measure the complexity of polynomials. In this section we give some definitions and well
known properties of arithmetic circuits. For more background see
e.g.\ \cite{bur00,malodportier}.

An \emph{arithmetic circuit} is a directed acyclic graph with vertices
of indegree $0$ or $2$ called \emph{gates}. The gates of indegree $0$
are called the \emph{inputs} and are labeled with elements of $\K$ or
variables. The gates of indegree $2$, called \emph{computation gates}, are
labeled with operations of the field ($+$ and $\times$). The
polynomial computed by a gate is defined inductively. The polynomial
computed by an input gate is the one corresponding to its label. The
polynomial computed by a computation gate is the sum or the product
of the polynomials computed by its children. We assume that there exists
a distinguished gate called the \emph{output}. The polynomial computed
by an arithmetic circuit is the one computed by its output gate. The
\emph{size} of a circuit $C$, denoted by $|C|$, is the number of
vertices of its underlying DAG.

An arithmetic circuit $C$ is said to be 
\emph{skew}
if for each $\times$-gate at least one of its children is an input of
the circuit. A circuit is said to be \emph{multiplicatively disjoint}
if for each $\times$-gate, its two input subcircuits are disjoint.

A family $(f_n)_{n \in \N}$ of polynomials is in $\VP$ if there exists
a family of multiplicatively disjoint circuits $(C_n)_{n \in \N}$ and
a polynomial $P$ such that for all $n \in \N$, $C_n$ computes $f_n$
and $|C_n| \leq P(n)$. The family $(f_n)$ is in 
$\VPw$ if the $C_n$ are skew.

\begin{remark}
Originally, $\VP$ was defined as families of polynomials that can be computed by
polynomial size circuits and have polynomially bounded degree. As
shown in \cite{malodportier}  the definition given here is
equivalent to the original one. We prefer this one 
here because
the semantic condition on the degree is harder to deal with than multiplicatively disjointness which is more syntactic.
\end{remark}

In the following, we will simulate arithmetic circuits by formulas
computing tensors. We 
use the notion of
\emph{parse trees} of a circuit. For a multiplicatively disjoint
circuit $C$, we define its parse trees inductively. A parse tree $T$
of $C$ is a subgraph of $C$ constructed as follows:
\begin{itemize}
\item Add the output of $C$ to $T$
\item For every gate $v$ added to $T$ do the following:
\begin{itemize}
\item If $v$ is a $+$-gate, add exactly one of its children to $T$.
\item If $v$ is a $\times$-gate, add both of its children to $T$.
\end{itemize}
\end{itemize}
As $C$ is multiplicatively disjoint, a parse tree of $C$ is a
tree. The monomial $m(T)$ computed by a parse tree $T$ is the product
of the labels of its leaves. The polynomial computed by $C$ is the sum
of the monomials of all parse trees of $C$.

\subsection{Tensors}
In this paper, we interpret tensors as multidimensional arrays. Their
algebraic nature is not studied here.
  A good introduction
to multilinear algebra and tensors can be found in
\cite{marv}.

Let $\tuple n k$ be $k$ positive integers. A \emph{$k$-dimensional
  tensor} $T$ of \emph{order} $(\tuple n k)$ is a mapping  $T:\set{n_1}
\times \ldots \times \set{n_k} \rightarrow \mathbb{K}$. For $i_1 \in
\set{n_1}, \ldots, i_k \in \set{n_k}$, we denote by $T[\tuple i k]$
the value of the mapping on the point $(\tuple i k)$. We call these
values \emph{entries} of $T$. We denote by~$D(T)$ the \emph{domain} of~$T$. Obviously~$D(T) = \set{n_1} \times \ldots \times \set{n_k}$.

The \emph{size} of a tensor $T$, denote by $\|T\|$, is the number of
entries, i.e. $\|T\| = \prod_{i=1}^k n_i$, where~$T$ is of order
$(\tuple n k)$. The \emph{maximal order} of $T$, denote by $\maxorder
T$ is $\max_{i \in \set{k}} n_i$.
In the following, we also call tensors of dimension $2$ matrices and 
tensors of dimension~$1$ vectors.

\begin{definition}[Contraction]
Let $T$ be a $k$-dimensional tensor of order $(\tuple n k)$ and $G$ an
$l$-dimensional tensor of order $(\tuple m l)$ with $k,l \geq 1$. If
$n_k = m_1$, we denote by $T * G$ the \emph{contraction} of $T$ and $G$ on the dimensions $k$ and $1$ which is a tensor of order $(\tuple n
{k-1}, \tuple[2] m l)$ defined as $(T * G)[\e_1, \e_2] =
\sum_{i=1}^{n_k} T[\e_1, i]G[i, \e_2]$ for all $\e_1\in \set{n_1} \times \ldots \times \set{n_{k-1}}$ and $\e_2 \in \set{m_2}\times \ldots \times \set{m_l}$. 
\end{definition}

\begin{remark}
Obviously, contraction is a direct generalization of the matrix
product. Indeed, if both $T$ and $G$ are matrices, then $T*G$ is the
ordinary matrix product.
\end{remark}

\begin{proposition}\label{prop:associativity}
 Let $T$, $G$, $H$ be tensors with $\dim(G)\ge 2$ such that $T * (G *
 H)$ and $(T * G) * H$ are both well defined. Then $$T * (G * H) = (T
 * G) * H.$$
\end{proposition}

\begin{proof}
By direct consequence of the associativity of the ordinary
matrix product. For a tensor $T$ of dimension $k \geq 2$ and order
$(\tuple n k)$, and 
a tuple $\e$ of length $k-2$ we define the
$n_1 \times n_k$ matrix $T_\e :=(T[i,\e,j])_{i \in \set{n_1},j \in
  \set{n_k}}$. Then by
  associativity of matrix product: 
$$ \forall \e_1,\e_2,\e_3\colon T_{\e_1} * (G_{\e_2} * H_{\e_3}) =
(T_{\e_1} * G_{\e_2}) * H_{\e_3}.$$ Hence, the claim follows when
$\dim(T) \geq 2$ and $\dim(H) \geq 2$. For $\dim(T) = 1$ or $\dim(H) =
1$ the argument is similar.
\end{proof}

\begin{observation}\label{obs:breakassociativity}
 If $G$ is a vector, then the equality of Proposition
 \ref{prop:associativity} may not be true anymore. For example
 $$ \left( \begin{pmatrix}0&1\\0&0\end{pmatrix}
   * \begin{pmatrix}0\\1\end{pmatrix}\right)
     * \begin{pmatrix}1&0\\0&0\end{pmatrix}
       = \begin{pmatrix}1\\0\end{pmatrix}
         \ne \begin{pmatrix}0\\0\end{pmatrix}
           = \begin{pmatrix}0&1\\0&0\end{pmatrix} *
             \left(\begin{pmatrix}0\\1\end{pmatrix}
               * \begin{pmatrix}1&0\\0&0\end{pmatrix}\right).$$
\end{observation}

\begin{definition}
A \emph{$\sform$} $F$ is a labeled, ordered, rooted binary tree whose
the leaves, called the inputs, are labeled by tensors whose entries
are elements of $\K$ or variables and the other nodes are labeled by
$*$.
The tensor $T_v$ computed by a node $v$ is defined inductively:
\begin{itemize}
\item If $v$ is a leaf then $T_v:=\lab v$.
\item If $v$ is labeled by $*$ and has left child $v_1$ and right child $v_2$ then
  $T_v := F_{v_1} * F_{v_2}$.
\end{itemize}
A $\sform$ computes the tensor computed by its root.
\end{definition}

As the entries of the input tensors are constants of $\K$ or
variables, each entries of a tensor computed by a gate is a polynomial
of $\K[\tuple X n]$. This is why it makes sense to compare the
computational power of $\sform$s and arithmetic circuits defined in
the last section. Moreover, of all the polynomial computed in the output of an $\sform$ we will mostly only be interested in one single polynomial. Thus we assume that the output tensor has only one single entry, i.e.\ the tensor is indeed a scalar. Observe that this form can always be achieved by contracting with vectors. We say that the scalar polynomial computed by a $\sform$ is \emph{the} polynomial computed by it.

\begin{definition}
The \emph{size} of a $\sform$ $F$, denoted by $|F|$, is the number of
$*$-gates plus the size of the inputs, i.e.\ $|F|:= |\{v\mid \lab{v}=
*\}| + \sum_{T: T \text{ input of }F} \|T\|$. The \emph{dimension} of
$F$, denoted by $\dim(F)$ is the dimension of the tensor computed by
$F$. The \emph{maximal dimension} of $F$, denoted by $\maxdim(F)$ is
the maximal dimension of the tensors computed at the gates of $F$,
i.e. $\maxdim(F) := \max_{v:\text{ gate in $F$}} (\dim(T_v))$. The
\emph{input dimension} of $F$ is $\max_{v: v\text{ input of } F}
\dim(T_v)$.
\end{definition}

We will often mix the notations for tensors and for tensor
formulas. For example, if $F$ is a tensor formula computing the tensor
$T$, we will speak of the order of $F$ instead of $T$ and write
$F[\e]$ instead of $T[\e]$. Moreover, given two different formulas $F$
and $F'$, we will write $F \simeq F'$ if they compute the same tensor.

\section{From arithmetic circuits to $\sform$s}

We  describe how a family of polynomials in $\VP$ can
be simulated by a family of $\sform$s of polynomial size and  maximal dimension
$3$. Our proof is inspired by a proof from \cite{men} where it is 
shown that polynomials in $\VP$ can be represented by bounded
treewidth CSPs. 

\begin{theorem}
\label{vpform}
Let $(f_n) \in \VP$. There exists a family of $\sform$s $(F_n)$ of
maximal dimension $3$ and polynomial size such that $F_n$ computes $f_n$ for all $n$.
\end{theorem}

We use the following observation from \cite{men} which can be proved by combining results from Malod and Portier
\cite{malodportier} and Valiant et al.\ \cite{valpar}.

\begin{proposition}\label{prop:isomorphic}
 Let $f$ be computed by an arithmetic circuit $C$ of size $s$. Then there is an arithmetic circuit $C'$ of size $s^{O(1)}$ that also computes $f$ such that all parse trees of $C'$ are isomorphic to a common tree $\calT$.
\end{proposition}

Theorem \ref{vpform} follows direcly from Proposition \ref{prop:isomorphic} and the following lemma:

\begin{lemma}
Let $C$ be an arithmetic circuit computing the polynomial $f$ whose
parse trees are all isomorphic to a common parse tree $\calT$. Then there
exists a $\sform$ $F$ of maximal dimension $3$
and of size $9|C|^3|\calT|$ that computes $f$.
\end{lemma}
\begin{proof}
We 
construct a tensor formula along the tree $\calT$ which contains the sum of all monomials of $f_n$ in its entries. 
We denote by $V(\calT)$ (resp. $V(C)$) the vertices of $\calT$ (resp. $C$). For $s \in V(\calT)$, we call $\calT_s$ the subtree of
$\calT$ rooted in $s$. We define a \emph{partial parse tree} rooted in
$s$ to be a function $p : V(\calT_s) \rightarrow V(C)$ respecting the
following conditions for all $t \in V(\calT_s)$:
\begin{enumerate}
\item If $t$ is a leaf, then $p(t)$ is an input of $C$.
\item If $t$ has one child $t_1$, $p(t)$ is a $+$-gate and $p(t_1)$ is
  a child of $p(t)$ in $C$.
\item If $t$ has two children $t_1$ and $t_2$, then 
\begin{enumerate}
\item $p(t)$ is a $\times$-gate,
\item $p(t_1)$ is the left child of $p(t)$, and
\item $p(t_2)$ is the right child of $p(t)$.
\end{enumerate}
\end{enumerate}

We call these conditions the parse tree conditions. It is easy to see
that when $s$ is the root of $\calT$ (and thus $\calT_s = \calT$) and
$p: V(\calT) \rightarrow V(C)$, then $p(V(\calT))$ is the vertex set of a parse tree of
$C$ if and only if $p$ is a partial parse tree rooted in
$s$.

If $p$ is a partial parse tree rooted in $s \in V(\calT)$, we define
the monomial $m(p)$ computed by $p$ by $m(p) := \prod_{t \in
  \leaf{\calT_s}} \lab{p(t)}$. Observe that this is well defined as $p$
respects, in particular, the first parse tree condition and thus
$p(t)$ for $t \in \leaf{\calT_s}$ is always an input of $C$. If $p$ does not respect
the parse tree conditions, we set $m(p) = 0$. With this notation we have
$$f = \sum_{p : V(\calT) \rightarrow V(C)} m(p).$$

We index the vertices of $C$ : $V(C) = \{\tuple v r\}$ with $r =
|C|$. We denote by $E$ the tensor of dimension $1$ and order $(r)$
such that for all $i \leq r$, $E[i] = 1$ and by $\delta_{i,j}$ the
dirac function which equals $1$ if $i=j$ and $0$ otherwise. We
construct by induction along the structure of $\calT$ a $\sform$~$F_s$ for each  $s \in \calT$. The formula $F_s$ has dimension
$2$, order $(r,r)$, size at most $9r^3|\calT_s|$ and maximal dimension
$3$.  Furthermore, for all $i,j \leq r$: $$ F_s[i,j] = \delta_{i,j}
\sum_{p : V(\calT_s) \rightarrow V(C) \atop p(s) = v_i} m(p).$$
Observing $f = E * F_s * E$ when $s$ is the root of $\calT$ completes
the proof. We now describe the inductive construction of
$F_s$. Several cases occur:

~ 

\noindent {\bf $s$ is a leaf: } In this case $\calT_s$ consists only of the leaf
$s$. The partial parse trees of $\calT_s$ are functions $p : \{s\}
\rightarrow V(C)$ and $m(p) = \lab{p(s)}$ if $p(s)$ is a input of $C$
and $m(p) = 0$ otherwise. Then $F_s$ 
consists of a $r \times r$ input matrix $I$ such that for all $i,j \leq r$,
$$I[i,j] = \left\{
    \begin{array}{ll}
      \delta_{i,j} \lab{v_j} \text{ if } v_j \text{ is an input} \\ 0 \text{
        otherwise.}
    \end{array}
\right.
$$ 
Obviously, $F_s$ is of size $r \leq 9r^3$, of maximal dimension $2$
and $$I[i,j] = \delta_{i,j} \sum_{p : \{s\} \rightarrow V(C) \atop
  p(s) = v_i} m(p)$$.

\noindent {\bf $s$ has one child $s_1$: } We start with an observation
on functions $p : V(\calT_s) \rightarrow V(C)$. Let 
$p_1$ be
the restriction of $p$ on $V(\calT_{s_1})$. If $p$ is a partial parse
tree, then $p_1$ is one, too, because it fulfills the parse tree
conditions for all $t \in V(\calT_{s_1}) \subseteq
V(\calT_s)$. Moreover, $m(p) = m(p_1)$ because the leaves in $p$
and in $p_1$ are the same. In addition, if $p$ is not a partial
parse tree then
\begin{itemize}
\item either $p$ violates a parse tree condition for $t \in
  \calT_{s_1}$. In that case, $p_1$ is not a partial parse tree and
  then $m(p) = m(p_1) = 0$,
\item or $p(s)$ is not a $+$-gate,
\item or $p(s)$ is a $+$-gate but $p(s_1)$ is not a child of
  $p(s)$.
\end{itemize}
We encode these conditions in a tensor of dimension $3$ and order
$(r,r,r)$ defined as
$$ M[i,j,k] := \left\{
\begin{array}{l}
\delta_{j,k} \text{ if } v_j \text{ is $+$-gate and } v_i \text{ is a
  child of } v_j \\ 0 \text{ otherwise. }
\end{array}
\right.
$$ 

Let $F_s$ be the formula $$F_s := E * (F_{s_1} * M)$$ of maximal
dimension $3$, dimension $2$ and order $(r,r)$. We have $|F_s| \le 9r^3(|\calT_s| - 1) + r^3 + 2 + \|E\| \leq
9r^3|\calT_s|$ and
$$ 
\begin{aligned}
F_s[j,k] &= \sum_{i=1}^r (\sum_{p=1}^r F_{s_1}[i,p]M[p,j,k]) \\ &= 
\delta_{j,k} \sum_{i=1}^r (\sum_{p_1 : V(\calT_{s_1}) \rightarrow V(C)
  \atop p_1(s_1) = v_i} m(p_1) M[i,j,j]).
\end{aligned}
$$

Let $p$ be a function $p : V(\calT_{s}) \rightarrow V(C)$ such that
$p(s_1) = v_i$ and $p(s) = v_j$. Let $p_1$ be its restriction on
$V(\calT_{s_1})$. We have
$m(p) = m(p_1)M[i,j,j]$, because
\begin{itemize}
\item if $p$ is a partial parse tree then $M[i,j,j] = 1$ and thus
  $m(p_1)M[i,j,j] = m(p_1) = m(p)$,
\item if $v_j$ is not a $+$-gate then $m(p) = 0$ and also 
  $m(p_1)M[i,j,j]=0$ because $M[i,j,j] = 0$,
\item if $p_1(s_1) = v_i$ is not a child of $v_j = p(s)$ then
  $m(p) = 0$. Since $M[i,j,j] = 0$, we have $m(p) = 0 =
  m(p_1)M[i,j,j]$.
\end{itemize}

Thus we have in each case
$$ F_s[j,k] = 
 \delta_{j,k} \sum_{i=1}^r \sum_{p_1 : V(\calT_{s_1}) \rightarrow V(C)
  \atop p_1(s_1) = v_i} m(p_1)M[i,j,k] = \delta_{j,k} \sum_{p : 
  V(\calT_s) \rightarrow V(C) \atop p(s)=   v_j} m(p)$$ which 
  completes the proof of this case.

\noindent {\bf $s$ has two children, $s_1$ (left child) and $s_2$ (right child): } 
As above, we encode the parse tree conditions in tensors of
dimension $3$ and contract them correctly to compute the desired result. This
time there are two different tensors: one encoding the condition 3.b
and one for~3.c. Let $M_L$ and $M_R$ be the two following $(r,r,r)$
tensors:

$$ M_L[i,j,k] = \left\{
\begin{array}{l}
\delta_{j,k} \text{ if } v_j \text{ is a $\times$-gate and } v_i
\text{ is the left child of } v_j \\ 0 \text{ otherwise,}
\end{array}
\right.
$$ 
$$ M_R[i,j,k] = \left\{
\begin{array}{l}
\delta_{j,k} \text{ if } v_j \text{ is a $\times$-gate and } v_i
\text{ is the right child of } v_j \\ 0 \text{ otherwise. }
\end{array}
\right.
$$ 

Let $F_s$ be the formula of maximal dimension $3$, dimension $2$ and
order $(r,r)$ defined as
$$ F_s = (E * (F_{s_1} * M_L)) * (E * (F_{s_2} * M_R)).$$

 We have
$|F_s| = |F_{s_1}|+\|M_L\|+\|M_R\|+|F_{s_2}|+5+2\|E\| \leq 9r^3|\calT_s|$.
In addition:
$$ 
\begin{aligned}
F_s[i,j] &= \sum_{k=1}^r ((\sum_{a=1}^r F_{s_1}[a,a]M_L[a,i,k])
\times (\sum_{b=1}^r F_{s_2}[b,b]M_R[b,k,j])) \\ &= \delta_{i,j}
\sum_{a,b = 1}^r F_{s_1}[a,a]M_L[a,i,i]F_{s_2}[b,b]M_R[b,i,i] \\ &= 
\delta_{i,j} \sum_{a,b = 1}^r \sum_{p_1 : V(\calT_{s_1}) \rightarrow
  V(C) \atop p_1(s_1) = v_a} \sum_{p_2 : V(\calT_{s_2}) \rightarrow
  V(C) \atop p_2(s_2) = v_b} m(p_1)m(p_2)M_L[a,i,i]M_R[b,i,i]
\end{aligned}.
$$

Similarly to before consider $p : V(\calT_{s})
\rightarrow V(C)$ such that $p(s) = v_i$, $p(s_1) = v_a$ and
$p(s_2)=v_b$. We denote by $p_1$ the restriction of $p$ on
$V(\calT_{s_1})$ and by $p_2$ its restriction on $V(\calT_{s_2})$.
We will show $$m(p) = M_L[a,i,i]M_R[b,i,i]m(p_1)m(p_2)$$ by
studying the possible cases:
\begin{itemize}
\item If $p$ is a partial parse tree then $p_1$ and $p_2$
  are, too. Moreover, since $s$ has two children, $p(s) = v_i$ is necessarily
  a $\times$-gate, $v_a$ its left child and $v_b$ its right
  child. It follows that $M_L[a,i,i] = M_R[b,i,i] = 1$ and $$m(p) = \prod_{l \in
    \leaf{\calT_s}} \lab{l} = \prod_{l \in \leaf{\calT_{s_1}}}
  \lab{l}\prod_{l \in \leaf{\calT_{s_2}}} \lab{l} = m(p_1)m(p_2).$$
\item If $p$ is not a partial parse tree then three cases can occur: If $p_1$
  (resp.\ $p_2$) is not a partial parse tree, then $m(p_1) = 0$
  (resp.\ $m(p_2) = 0$). If $v_i$ is not a $\times$-gate, then
  $M_L[a,i,i] = 0$. Finally, if $v_a$ (resp. $v_b$) is not the left
  (resp.\ right) child of $v_i$, then $M_L[a,i,i] = 0$
  (resp.\ $M_R[b,i,i] = 0$). In all those cases,
  $M_L[a,i,i]M_R[b,i,i]m(p_1)m(p_2) = 0 = m(p)$.
\end{itemize}

This completes the proof.
\end{proof}

\section{From $\sform$s to arithmetic circuits}\label{sct:upper}

In this section we will show that the polynomials computed by
polynomial size $\sform$s can also be computed by polynomial size arithmetic
circuits. We start by first proving this for formulas with bounded
maximal dimension. Then we extend this result by showing that any
$\sform$ can be transformed into an equivalent one with bounded
maximal dimension without increasing the size.

\subsection{Formulas with bounded maximal dimension}

\begin{proposition}
Let $F$ be a $\sform$ of maximal dimension $k$, dimension $l \leq k$
 and order $(\tuple n l)$. 
Let $n := \max_{T: T \text{ input of }F}
(\maxorder{T})$. Then there exists a multiplicatively disjoint circuit
$C$ of size at most $2n^{k+1}|F|$ such that for all $\e \in D(F)$ there exists a gate $v_\e$ in $C$ computing $F[\e]$.
\end{proposition}

\begin{proof}

If $F$ is an input, let $C$ be the circuit having $\prod_{i=1}^l n_i$
inputs, each one labeled with an entry of $F$. The size of $C$ is
$\prod_{i=1}^l n_i \leq n^k$.

If $F = G * H$, by induction we have circuits $C_G$ and $C_H$ with the
desired properties for~$G$ and~$H$. The dimension of $F$ is less than
$k$ and for $\e \in D(F)$, $F[\e] = \sum_{i=1}^m G[\e_1, i]H[i, \e_2]$
with $m \leq n$.

Each $G[\e_1,i]$ and $H[i,\e_2]$ is computed by a gate of $C_G$ and
$C_H$,respectively, so we can compute $F[\e]$ by adding at most $2n$
gates ($m$ $\times$-gates and $m-1$ $+$-gates). As there are at most
$n^k$ entries in $F$, we can compute all of them with a circuit $C$ by adding at
most $2n \times n^k$ gates to $C_H \cup C_G$.

The circuit $C$ is multiplicatively disjoint since each $\times$-gate
receives one of its input from $C_G$ and the other one from $C_H$.
Also $|C| = |C_G| + |C_H| +
2n^{k+1} \leq 2n^{k+1}|F|$.
\end{proof}

\begin{corollary}
\label{formvp}
Let $(F_n)$ be a family of $\sform$s of polynomial size and of maximal
dimension $k$ computing a family $(f_n)$ of polynomials.
Then $(f_n)$ is in $\VP$.
\end{corollary}

\subsection{Unbounded maximal dimension}

Since the size of the circuit constructed in the previous section is exponential in $k:=\maxdim(F)$, we cannot apply the results from there directly
if $k$ is not bounded by a constant. Somewhat
surprisingly we will see in this section that one does not gain any
expressivity by letting intermediate dimensions of formulas
grow arbitrarily. Thus bounding $\maxdim(F)$ is not a restriction of the computational power of $\sform$s.

\begin{definition}
A $\sform$ $F$ of dimension $k$ and input dimension $p$ is said to be
\emph{tame} if $\maxdim(F) \le \max(k,p)$.
\end{definition}

\begin{definition}
A $\sform$ $F$ is said to be \emph{totally tame} if each subformula of $F$ is
tame.
\end{definition}

Let us remark again that also in \cite{mck} and \cite{mal} there is a notion of tameness that prevents intermediate results from growing too much during the computation. It turns out that in those papers tameness plays a crucial role: Tame formulas can be evaluated efficiently while general formulas are hard to evaluate in the respective models. We will see that in our setting tameness is not a restriction at all. Indeed, any $\sform$ can be turned into an equivalent totally tame formula without any increase of its size. Thus totally tame and general formulas have the same expressive power in our setting which is a striking difference to the setting from \cite{mck} and \cite{mal}.
We start with the following lemma:

\begin{lemma}
\label{leftdecompose}
Let $F$ be a totally tame formula with $\dim(F)=k$ and input dimension~$p$. 
For all totally tame formulas $E$ of dimension $1$ and input
dimension at most $p$, there exist totally tame formulas $G_r$ and $G_l$ of size $|F *
E| = |E * F|= |F| + |E| +1$ such that $G_r \simeq F * E$ and $G_l \simeq E * F$.
\end{lemma}

\begin{proof}
We only show the construction of $G_r$; the construction of $G_l$ is completely analogous. We proceed by induction on $F$.

If $F$ is an input, then $\maxdim(F)=\dim(F) = p$. Let $E$ be any totally tame formula of dimension
$1$ and input dimension at most $p$. We set $G_r := F * E$. Clearly, $k=\dim(G_r)= p-1$. Furthermore, $\maxdim(G_r) = \max(p-1, \maxdim(F), \maxdim(E))\le p$ because $E$ has input dimension at most $p$ and is totally tame. Thus $G_r$ is totally tame.

Let now $F = F_1 * F_2$. Let $k_1:=\dim(F_1)$ and $k_2:=\dim(F_2)$. Let $E$ be a
totally tame formula of dimension $1$ and input dimension at most $p$.
\begin{itemize}
\item If $\dim(F_2) = 1$, we claim that $G_r = F * E$ is totally tame. Indeed, since
  $\dim(F_2)=1$, we have $\dim(F)=k_1-1$. But $F$ is by assumption tame, so $k_1 = \dim(F_1) \leq
  \max(k_1-1,p)$. Hence $k_1 \leq p$ and $\dim(F) \leq p$. Thus all intermediate results of $G_r$ have dimension at most $p$, so $G_r$ is tame. But then it is also totally tame, because its subformulas are totally tame by assumption. 
\item If $\dim(F_2) = 2$, we have $p\ge 2$ obviously and $\dim(F)=\dim(F_1)$. $F_2$ is a subformula of~$F$, so it is totally tame, too. Furthermore, $F_2 * E$ is of dimension $1$ and it is also totally tame since $2 \leq p$. Moreover, by Proposition \ref{prop:associativity} we have $F * E \simeq F_1 * (F_2 * E)$. Applying the induction hypothesis on $F_1$ and $(F_2 * E)$ gives the desired $G_r$.
\item If $\dim(F_2) > 2$, by Proposition \ref{prop:associativity} we have $F * E \simeq F_1
  * (F_2 * E)$. We first apply the induction hypothesis on $F_2$ and $E$ to
  construct a totally tame formula $G'$ computing $F_2 * E$. Finally
  $G_r := F_1 * G'$ is totally tame since $F_1$ and $G'$ are totally tame
  and $F$ is of dimension $k = k_1 + k_2 - 2 \geq \max(k_1,k_2-1)$.
\end{itemize}
\end{proof}
We now prove the main proposition of this section:
\begin{proposition}
\label{boundedlength}
For every $\sform$ $F$ there exists a totally tame $\sform$ $F'$ such
that $F' \simeq F$ and $|F| = |F'|$.
\end{proposition}

\begin{proof}
Proof by induction on $F$.

If $F$ is an input then it is trivially totally tame as the
dimension of $F$ is equal to the input dimension of $F$. So we set 
$F' := F$.

If $F = F_1 * F_2$ then several cases can occur depending on the
dimension of $F_1$ and $F_2$. We denote by $k,k_1,k_2$ the dimensions
of $F$, $F_1$ and $F_2$ respectively. We recall that $k = k_1 + k_2 -
2$.
\begin{itemize}
\item If both $k_1$ and $k_2$ are different from $1$. Then $F' = F_1' *
  F_2'$ is totally tame since $k \geq \max(k_1,k_2)$
\item If $k_1=1$ or $k_2 = 1$, we use Lemma $\ref{leftdecompose}$ on $F_1'$ and
  $F_2'$ to construct $F'$ of size $|F|$, totally tame, computing $F_1
  * F_2$.
\end{itemize}
\end{proof}

Combining Proposition \ref{boundedlength} and Corollary \ref{formvp} we get the following theorem:

\begin{theorem}
\label{vpgen}
Let $(F_n)$ be a family of $\sform$s of polynomial size and input
dimension $p$ (independent of $n$) computing a family of polynomials $(f_n)$. Then  $(f_n)$ is in $\VP$.
\end{theorem}
\begin{proof}
Applying Proposition \ref{boundedlength} on $(F_n)$ gives a family $(F_n')$ computing $(f_n)$ such that $(F_n')$ is tame. Then the maximal dimension of $F_n'$ is $p$ (because $F_n$ is scalar, thus of dimension $1$) and applying Corollary \ref{formvp} proves the claim.
\end{proof}

\subsection{Unbounded input dimension}

While we got rid of the restriction on the maximum dimension of
$\sform$s in the last section, we still have a bound on the dimension
of the inputs in Theorem \ref{vpgen}. In this section we
will show that this bound is not necessary to have containment of the
computed polynomials in $\VP$. We will show that inputs having ``big''
dimension can be computed by polynomial size $\sform$s of input
dimension $3$. We can then use this to eliminate inputs of dimension
more than $3$ in $\sform$s. Applying Theorem \ref{vpgen} we
conclude that the only restriction on $\sform$s that we need to ensure
containment in $\VP$ is the polynomial size bound.

\begin{proposition}
\label{unbounded}
Let $T$ be a $r$-dimensional tensor of order $(\tuple n r)$. Let $L :=
\|T\| = \prod_{i=1}^r n_i$ be the number of entries in $T$. Then there
is a $\sform$ $F$ of size $r+1+L^3+2L$ and input dimension $3$
computing $T$.
\end{proposition}

\begin{proof}[Proof (sketch)]
Choose an arbitrary bijection $B:[L] \rightarrow \set{n_1} \times \ldots \times 
\set{n_r}$. Let furthermore $B_i : \set{L} \rightarrow \set{n_i}$ for
$i \leq r$ be the projection of $B$ onto the $i$-th
coordinate.
We define the $3$-dimensional tensors $T_i$ of order
$(L,n_i,L)$ by

$$ 
T_1[m, k, n] = \left\{
    \begin{array}{ll}
      T[B(m)] \text{ if } m = n \text{ and } B_1(m) = k \\ 0 \text{
        otherwise.}
    \end{array}
\right.
$$ and,  for $2 \leq i \leq r$, 
$$ 
T_i[m, k, n] = \left\{
    \begin{array}{ll}
      1 \text{ if } m = n \text{ and } B_i(m) = k \\ 0 \text{
        otherwise.}
    \end{array}
\right.
$$ 

By induction one can show that for the tensor $P = T_1 * \ldots * T_r$ we have
that 

$$
P[m,\tuple k r, n] = 
\left\{
    \begin{array}{l}
      T[\tuple k r]  \text{ if } m = n \text{ and } B(m) = (\tuple k r) \\ 
      0  \text{ otherwise.}
    \end{array}
\right.
$$

Hence $T = E * P * E$ where $E$
is a vector of order $L$ filled with $1$. The
complete proof is given in the appendix.
\end{proof}

The following theorem is a direct consequence of Proposition \ref{unbounded} and Theorem \ref{vpgen}.
\begin{theorem}\label{mostgeneral}
Let $(F_n)$ be a family of $\sform$s of polynomial
size computing a family of polynomials $(f_n)$. Then $f_n$ is in $\VP$.
\end{theorem}

\section{The power of contracting with vectors}

In this section we will make a finer examination of where exactly the
additional power originates when going from
iterated matrix product of \cite{malodportier} to tensor
contractions. We will see that this additional expressivity crucially
depends on the possiblity of contracting tensors on more than two of
their dimensions. We will show that when we prevent this
possibility by disallowing contractions with vectors -- which are used
in the proof of Theorem \ref{vpform} to ``collapse'' dimensions not
needed anymore so that we can access other dimensions to contract on -- the
expressivity of $\sform$s drops to that of iterated matrix product.

Observe that we cannot assume that $\sform$s compute scalars in this setting, because we cannot decrease the dimension of the tensors computed by a formula. Also we cannot compute all entries of the output at the same time efficiently, because those might be exponentially many ones. But we will see in the following Propositions that we can compute each individual entry of the output more efficiently than in the general setting where contraction with tensors is allowed.

\begin{proposition}\label{prop:novectors}
Let $F$ be a $\sform$ of order $(\tuple n k)$ whose inputs are all of
dimension at least $2$. Then for all $\e\in D(F)$ there exists a skew arithmetic circuit
$C$ of size at most $2n^3|F|$ where $n := \max_{T\colon T \text{ input of }F}
(\maxorder T)$ computing $F[\e]$.
\end{proposition}
\begin{proof}
By Proposition \ref{prop:associativity} we can write $F$ as $A_1 *
(A_2 * (A_ 3 * \ldots * A_n))$. The proof then follows easily by
induction: We do the same construction as in Theorem \ref{formvp} but
this time we only have $n^2$ entries and at each $*$-gate, one side is
an input, resulting in a skew circuit.
\end{proof}

The case of Proposition \ref{prop:novectors} exactly corresponds to
the characterization of $\VPw$ by Malod and Portier
\cite{malodportier} by $n$ products of matrices of size $n \times
n$. Thus Proposition \ref{prop:novectors} naturally generalizes this
result and the real new power seen in Theorem \ref{vpform} must come
from the use of vectors in the products.  As we have seen in the proof
of Proposition \ref{prop:novectors} it is crucial that vectors are the
only case which breaks the associativity of Proposition
\ref{prop:associativity}. So what looked like a not very important
edge case in Observation \ref{obs:breakassociativity} plays a
surprisingly important role for the expressivity of $\sform$s.

\section{The $*_{i,j}$ operators}

Our characterization of $\VP$ by $\sform$s contracts on dimension only in a very specific way in the contraction of two tensors: We always only contract on the last dimension of one tensor and the first dimension of the other one. It is thus very natural to ask if this is a restriction of the computational power of the formulas. In this section we will see that it is indeed not. If we allow free choice of the dimensions to contract on during a contraction this does not make the resulting polynomials harder to compute. To formalize this we give the folowing definition of a contraction $*_{i,j}$.

\begin{definition}
Let $T$ be a $k$-dimensional tensor of order $(\tuple n k)$ and $G$ a
$l$-dimensional tensor of order $(\tuple m l)$ with $k,l \geq 1$. When
$n_i = m_j$ for $i \leq k$ and $j \leq l$, we denote by $T *_{i,j} G$ the contraction of $T$ and $G$ on the dimensions $i$ and $j$
the $(k+l-2)$-dimensional tensor of order $(\tuple n {i-1},
\tuple[i+1] n k, \tuple m {j-1}, \tuple[j+1] m l)$ 
defined as $$(T *_{i,j} G)[\e_1,\e_2,\e_3,\e_4] = \sum_{r=1}^{n_i}
T[\e_1,r,\e_2]G[\e_3,r,\e_4]$$ for all $\e_1\in \set{n_1}
\times \ldots \times \set{n_{i-1}}$, $\e_2\in \set{n_{i+1}}
\times \ldots \times \set{n_k}$, $\e_3\in\set{m_1} \times
\ldots \times \set{m_{j-1}}$ and $\e_4\in\set{m_{j+1}}
\times \ldots \times \set{m_l}.$

\emph{$\sijform$s} are defined in complete analogy to $\sform$s.
\end{definition}

It turns out that $\sijform$s cannot compute more than $\sform$s, so the free choice of the dimensions to meld on does not change much.

\begin{theorem}
\label{thm:VPboundeddim}
Let $(F_n)$ be a family of $\sijform$s of polynomial
size computing a family of polynomials $(f_n)$.
Then $(f_n)$ is in $\VP$.\end{theorem}

The proof of Theorem \ref{thm:VPboundeddim} follows a similar approach as that of Theorem \ref{mostgeneral} and is thus given in the appendix for lack of space. Let us sketch some key steps here: If we bound the maximal
dimension of $\sijform$s by a constant $k$, it is easy to see that the proof of
Theorem \ref{formvp} can be adapted to $\sijform$s in a straightforward way. The main complication is then turning general $\sijform$s into totally tame ones. $*_{i,j}$ is not associative anymore, and this makes a straightforward translation of the proof of Proposition \ref{boundedlength} tricky. These problems can be solved by the observation that the crucial steps in the process of making a formula tame are those where a $\sijform$ is multiplied by a tensor of dimension $1$. But for such contractions we can give explicit formulas for different cases that may occur, so again every $\sijform$ has an equivalent tame $\sijform$. Combining this with Proposition \ref{unbounded} completes the proof.

\section{Conclusion}

We have shown that one can get a robust characterization of $\VP$ by formulas with tensors as input and tensor contraction as the only operation. This generalizes the known characterization of $\VPw$ by iterated matrix product by Malod and Portier \cite{malodportier}. In some aspects the situation in our setting is more subtle, though. We remarked that vectors and in general breaking associativity plays a crucial role if we want to characterize $\VP$. Also, unlike for iterated matrix product we have to make a choice if we take $*_{i,j}$ or $*$ as our basic operation. It is easy to check that that using the equivalent to $*_{i,j}$ for matrix product would merely be transposing the matrix, so it clearly does not change the expressivity of the model. But fortunately also in our setting, the choice of $*_{i,j}$ or $*$ does not influence the complexity of the computed polynomials.

Unfortunately, unlike for iterated matrix product our characterization seemingly does not directly lead to a characterization of $\VP$ by something similar to branching programs. We still think that such a characterization is highly desirable, because the branching program characterization of $\VPw$ has been the source of important insights in arithmetic circuit complexity. Thus we believe  that a similar characterization of $\VP$ might lead to a better understanding of $\VP$, a class that is arguably not very well understood, yet.

Let us quickly discuss several extensions to the results in this paper that we had to leave out for lack of space: First, analyzing the proofs of Section \ref{sct:upper} a little more carefully one can see that our results remain true if one does not measure the size of a tensor as the number of its entries but as the number of its \emph{nonzero} entries. This makes it possible to allow inputs of large dimension and large order.

Also, it seems plausible and straightforward to generalize our results to arbitrary semi-rings in the style of Damm, Holzer and McKenzie \cite{mck}. Choosing different semi-rings one would then probably get characterizations of classes like $\LOGCFL$ and its counting, mod-counting and gap-versions. The main new consideration would be the treatment of uniformity in these settings which appears to be possible with a more refined analysis of our proofs.

Finally, for tensors there are other natural operations to perform on them like addition or tensor product. It is natural to ask, if adding such operations does change the complexity of the resulting polynomials. While it is straightforward to see that adding only tensor product as an operation does not increase the expressivity of $\sform$s, we could so far not answer the corresponding question for addition. Therefore, we leave this as an open question.

\paragraph*{Acknowledgements}

We are very grateful for the very detailed and helpful feedback by Yann Strozecki on an early version of this paper. We would also like to thank Herv\'{e} Fournier, Guillaume Malod and Sylvain Perifel for helpful discussions.

\bibliography{biblio}

\newpage

\begin{appendix}

\section{Proof of proposition \ref{unbounded}}
Choose an arbitrary bijection $B:[L] \rightarrow \set{n_1} \times \ldots \times 
\set{n_r}$. Let furthermore $B_i : \set{L} \rightarrow \set{n_i}$ for
$i \leq r$ be the projection of $B$ onto the $i$-th
coordinate.

We define the $3$-dimensional tensors $T_i$ of order
$(L,n_i,L)$ by

$$ 
T_1[m, k, n] = \left\{
    \begin{array}{ll}
      T[B(m)] \text{ if } m = n \text{ and } B_1(m) = k \\ 0 \text{
        otherwise.}
    \end{array}
\right.
$$ and 
$$ 
T_i[m, k, n] = \left\{
    \begin{array}{ll}
      1 \text{ if } m = n \text{ and } B_i(m) = k \\ 0 \text{
        otherwise.}
    \end{array}
\right.
$$ for $i>1$.

We prove by induction that for the tensor $P_j := T_1 * \ldots * T_j$, $j\le r$, we have
that $P_j[m,\tuple k j, n] = T[\tuple k r]$ if $m = n$ and $B(m) =
(\tuple k j)$ and $P_j[m,\tuple k j, n]=0$ otherwise. 

For $j = 1$ is obvious by definition of $T_1$. So assume that it is true for $j-1$ then
$$
\def\arraystretch{1.5}
\begin{array}{lll}
P_j[m, \tuple k j, n] & = & \sum_{p = 1}^L P_{j-1}[m, \tuple k {j-1},
  p]T_j[p, k_j, n] \\  & = & P_{j-1}[m, \tuple k {j-1}, m]T_j[m, k_j,
  n] \\ & = &
\left\{
    \begin{array}{ll}
      T[B(m)] & \text{ if } m = n \text{ and } B_i(m) = k_i \text{ for } i
      \leq j-1 \\ & \text{ and } B_j(m) = k_j \\ 0 & \text{ otherwise.} 
    \end{array}
\right.
\end{array}
$$ This concludes the induction. 

Thus we have $P_r[m, \tuple k r, n] =
T[\tuple k r]$ if $B(m) = (\tuple k r)$ and $m = n$. We now sum over all $n,m\in [L]$ by $F = E * P_r * E$ where $E$ is a vector of size $L$ containing only $1$s. The resulting $\sform$ is of size $r+1+L^3+2L$
because each $T_i$ is of size $L^2n_i \leq L^3$ and $E$ is of size $L$
and we use $r+1$ $*$-gates.

\section{Proof of Theorem \ref{thm:VPboundeddim}}

The proof of Theorem \ref{thm:VPboundeddim} follows a similar approach as that of Theorem \ref{mostgeneral}. Let us first observe that an analogous version of Theorem \ref{formvp} with $*_{i,j}$ instead of $*$ can be proved easily. Also Proposition \ref{unbounded} applies directly for $*_{i,j}$.

Thus the only thing left to prove is that every $\sijform$ can be turned into a totally tame one. $*_{i,j}$ is not associative which makes a straightforward translation of the proof of Proposition~\ref{boundedlength}~tricky.
Still it is possible observing that for the
crucial case of contraction by a vector, it is possible to
prove identities that we can use where we applied associativity before.

\begin{proposition}\label{prop:ijidentities}
Let $F_1$ and $F_2$ be two tensors of dimension $k_1$ and $k_2$,
respectively. Let $E$ be a tensor of dimension $1$. Then 
$$ (F_1 *_{i,j} F_2) *_{k,1} E = 
\left\{
    \begin{array}{ll}
      F_1 *_{i,j} (F_2 *_{k-k_1+1,1} E) & \text{ if } k_1+j-1 \leq k \\
      F_1 *_{i,j-1} (F_2 *_{k-k_1+1,1} E) & \text{ if } k_1 \leq k \leq k_1+j-2 \\
      (F_1 *_{k,1} E) *_{i,j} F_2  & \text{ if } i \leq k < k_1 \\
      (F_1 *_{k,1} E) *_{i-1,j} F_2  & \text{ if } 1 \leq k < i \\
    \end{array}
\right.
$$
\end{proposition}
\begin{proof}
The proof follows from simple calculation. We denote by $(\tuple n {k_1})$ and $(\tuple
m {k_2})$ the orders of $F_1$ and $F_2$, respectively. Let $F$ be the
product $F = F_1 *_{i,j} F_2$. We have $$F[\e_1,\e_2,\e_3,\e_4] =
\sum_{r=1}^{n_i} F_1[\e_1,r,\e_2]F_2[\e_3,r,\e_4]$$ where $\e_1$ is a
tuple of length $i-1$, $\e_2$ of length $k_1-i$, $\e_3$ of length
$j-1$ and $\e_4$ a tuple of length $k_2-j$. 

If $k_1+j-1 \leq k$, then
$$
\begin{aligned}
(F * E)[\e_1,\e_2,\e_3, \f, \f'] &= \sum_{p=1}^{m}
  F[\e_1,\e_2,\e_3,\f,p,\f']E[p] \\ &= \sum_{p=1}^{m}
  \sum_{r=1}^{n_i}(F_1[\e_1,r,\e_2]F_2[\e_3,r,\f,p,\f'])E[p] \\ & =
  \sum_{r=1}^{n_i} \sum_{p=1}^m F_1[\e_1, r,
    \e_2](F_2[\e_3,r,\f,p,\f']E[p]) \\ &= F_1 *_{i,j} (F_2 *_{k-k_1+1,1}
  E)
\end{aligned}
$$ where $\f, \f'$ are the suitable subtuples of $\e_4$. The other cases can be checked in the same way. The only
difference is the position of $\f$ and $\f'$.
\end{proof}
Using these identities, we prove the following
lemma.
\begin{lemma}
\label{ijdecompose}
Let $F$ be a totally tame $\sijform$ of dimension $k$ and input
dimension $p$. For all totally tame formula $E$ of dimension $1$ and
for all $i \leq k$, there exists a totally tame $\sijform$ $G$ of size
$|F *_{i,1} E|$ such that $G \simeq F *_{i,1} E$.
\end{lemma}

\begin{proof}
The proof is done by induction on $F$.

If $F$ is an input, let $E$ be any totally tame formula of dimension
$1$ and $i \leq k$. Let $G = F *_{i,1} E$. Then $G$ is obviously totally
tame since $F$ and $E$ are totally tame and $G$ is of dimension $p-1$
so all the intermediate tensors are of dimension at most $p$.

If $F = F_1 *_{i,j} F_2$, then several cases can occur. Let $E$ be any
totally tame formula of dimension $1$ and $l \leq k$. We denote by
$k_1$ and $k_2$ the dimensions of $F_1$ and $F_2$ respectively. We
want to compute $F *_{l,1} E$. We proceed differently depending on $l$.

If $k_1 \leq l$, then we know, using the identities from Proposition \ref{prop:ijidentities} that
either $F *_{l,1} E \simeq F_1 *_{i,j} (F_2 *_{l-k_1+1,1} E)$ or $F
*_{l,1} E \simeq F_1 *_{i,j-1} (F_2 *_{l-k_1+1,1} E)$. So in both
cases, there exists $j'$ such that $F *_{l,1} E \simeq F_1 *_{i,j'}
(F_2 *_{l-k_1+1,1} E)$. Depending on the dimension of $F_2$, we again consider different cases:
\begin{itemize}
\item First of all, the case $\dim(F_2) = 1$ is not possible because
  if $k_2 = 1$ then $k = k_1 - 1$ and then necessarely $l < k_1$
\item If $\dim(F_2) = 2$ then $F_2 *_{l-k_1+1,1} E$ is a formula of
  dimension $1$. Moreover, it is totally tame because $F_2$ is totally
  tame since $F$ is, $E$ is totally tame by assumption and $\dim(F_2 *_{l-k_1+1,1} E)=1<p$. Thus we can apply the induction hypothesis on $F_1$ and the vector $F_2
  *_{l-k_1+1,1} E$
\item If $\dim(F_2) > 2$ then we first construct $G'$ by applying the
  induction on $F_2$ and $E$. The formula $G = F_1 *_{i,j'} G'$
  computes $F$ and is totally tame because both $F_1$ and $G'$ are
  totally tame and that $\dim(F) = k_1 + k_2 - 2 \geq \max(k_1,
  k_2-1)$ since $k_2 > 2$
\end{itemize}
This completes the case $k_1 \leq l$. 

We proceed similarly for the case $l < k_1$ using the other identities from Poposition \ref{prop:ijidentities}. In this case we have that there exists $i'$ such that $F *_{l,1}
E \simeq (F_1 *_{l,1} E) *_{i',j} F_2$. Again, we analyse depending on the
dimension of $F_1$.
\begin{itemize}
\item Again, since $l < k_1$, we have $\dim(F_1) \neq 1$.
\item If $\dim(F_1) = 2$ then $(F_1 *_{l,1} E)$ is a formula of
  dimension $1$. So $i'=1$ and thus $F \simeq F_2 *_{j,1} (F_1 *_{l,1}
  E)$. As before, $F_1 *_{l,1} E$ is totally tame and we apply the
  induction on $F_2$ and $F_1 *_{l,1} E$
\item If $\dim(F_1) > 2$ then first construct $G'$ by applying the
  induction on $F_1$ and $E$. The formula $G = G' *_{i',j} F_2$
  computes $F$ and is totally tame because both $F_2$ and $G'$ are
  totally tame and that $\dim(F) = k_1 + k_2 - 2 \geq \max(k_1,
  k_2-1)$ since $k_2 > 2$.
\end{itemize} 
\end{proof}

We now prove that $\sijform$s can always be turned into equivalent totally tame $\sijform$s.
\begin{proposition}
\label{ijboundedlength}
Let $F$ be an $\sijform$. Then there exists a totally tame $\sijform$~$F'$ such
that~$F' \simeq F$ and $|F| = |F'|$.
\end{proposition}
\begin{proof}
The proof is done by straightforward induction on $F$.

If $F$ is an input then it is trivially totally tame as the
dimension of $F$ is equal to the input dimension of $F$. We simply set $F' := F$.

If $F = F_1 *_{i,j} F_2$ then two cases can occur depending on
the dimension of $F_1$ and $F_2$. We denote by $k,k_1,k_2$ the
dimensions of $F$, $F_1$ and $F_2$, respectively. We recall that $k =
k_1 + k_2 - 2$.
\begin{itemize}
\item If both $k_1$ and $k_2$ are greater than $1$, then $F' = F_1'
  *_{i,j} F_2'$ is totally tame since $k \geq \max(k_1,k_2)$.
\item If $k_2 = 1$ or $k_1 = 1$, we use Lemma $\ref{ijdecompose}$ on
  $F_1'$ and $F_2'$ to construct $F'$ of size $|F|$, totally tame,
  computing $F_1 * F_2$.
\end{itemize}
\end{proof}

This completes the proof of Theorem \ref{thm:VPboundeddim}.
\end{appendix}
\end{document}